\documentclass[journal,draftcls,onecolumn,12pt,twoside,english]{IEEEtranTCOM}
\usepackage[T1]{fontenc}
\usepackage[fleqn]{amsmath}
\usepackage{graphicx}
\usepackage{amssymb}
\usepackage{amsfonts}
\usepackage{amscd}
\usepackage{multirow}
\usepackage{color}
\usepackage{babel}
\usepackage{subfigure}
\usepackage{algorithm}
\usepackage{algorithmicx}
\usepackage{listings}
\usepackage{epsfig}
\usepackage{algpseudocode}
\usepackage{array}
\usepackage{amsmath}
\usepackage{perpage}
\MakePerPage{footnote}
\newtheorem{theorem}{Theorem}
\newtheorem{lemma}[theorem]{Lemma}

\providecommand{\U}[1]{\protect\rule{.1in}{.1in}}
\makeatletter

\renewcommand{\figurename}{Fig.}
\addto\captionsenglish{\renewcommand{\figurename}{Fig.}}
\@ifundefined{definecolor}
{
}{}

\newtheorem{rem}{Remark}[section]

\makeatother
\begin{document}
\IEEEoverridecommandlockouts

\title{Energy Harvesting Broadband Communication Systems with Processing Energy Cost}
\author{\IEEEauthorblockN{Oner Orhan, Deniz G{\"u}nd{\"u}z,\emph{ Senior Member, IEEE,} and Elza Erkip,\emph{ Fellow, IEEE}}

}

\maketitle

\begin{abstract}
Communication over a broadband fading channel powered by an energy harvesting transmitter is studied. Assuming non-causal knowledge of energy/data arrivals and channel gains, optimal transmission schemes are identified by taking into account the energy cost of the processing circuitry as well as the transmission energy. A constant processing cost for each active sub-channel is assumed. Three different system objectives are considered: 1) throughput maximization, in which the total amount of transmitted data by a deadline is maximized for a backlogged transmitter with a finite capacity battery; 2) energy maximization, in which the remaining energy in an infinite capacity battery by a deadline is maximized such that all the arriving data packets are delivered; 3) transmission completion time minimization, in which the delivery time of all the arriving data packets is minimized assuming infinite size battery. For each objective, a convex optimization problem is formulated, the properties of the optimal transmission policies are identified, and an algorithm which computes an optimal transmission policy is proposed. Finally, based on the insights gained from the offline optimizations, low-complexity online algorithms performing close to the optimal dynamic programming solution for the throughput and energy maximization problems are developed under the assumption that the energy/data arrivals and channel states are known causally at the transmitter.
\end{abstract}


\section{Introduction}
\label{introduction}
Wireless sensor nodes are typically designed to have low cost and small size. These design objectives impose restrictions on the capacity and efficiency of the energy storage units that can be used. As a result, continuous operation of the sensor network requires frequent battery replacements, which increases the maintenance cost. Energy harvesting (EH) devices are able to overcome these challenges by collecting energy from the environment. However, due to the nature of the ambient energy sources, the amount of useful energy that can be harvested is limited and unreliable. Consequently, optimal management of the harvested energy becomes a new challenge for EH wireless nodes.

In most communications literature the energy cost of operating transmitter circuitry, such as digital-to-analog converters, mixers, filters, etc. is ignored. In short range communications, as in most wireless sensor networks, where inter-node distances are less than 10m, processing energy consumption can be comparable to the transmission energy \cite{bahai}. When the processing cost is negligible, increasing the transmission time and lowering the transmission power increases the energy efficiency (nats-per-joule), provided the rate-power function is monotonically increasing and concave, properties satisfied by most common transmission schemes as well as Shannon's capacity function. However, as shown in \cite{glue}, when processing cost is taken into account, bursty transmissions separated by ``sleep'' periods become optimal. In EH communication systems, this affects the optimal power allocation scheme considerably since both the power allocation and the sleep intervals will depend on the energy arrival profile.

In this paper, we consider an EH transmitter with processing cost communicating over a broadband fading channel, modelled as $K$ parallel sub-channels with each sub-channel having independent fading. Following the power consumption model in \cite{glue} and \cite{Li}, processing energy cost is modelled as a function of the transmission bandwidth and time and is assumed to be equal to a constant value for each sub-channel. We characterize optimal transmission policies for three different system objectives under the offline optimization framework which assumes that all channel gains and the sizes of arriving energy and  data packets are known non-causally before transmission starts. First, we only consider energy packet arrivals over time for a backlogged transmitter\footnote{A backlogged transmitter is the one that always has data available for transmission.} with a finite capacity battery, and we maximize the amount of total data delivered by a deadline $T$. We call this the \emph{throughput maximization problem} \cite{finite}. Throughput maximization is an important objective for high data rate applications. Then, we consider both data and energy arrivals over time and an unlimited battery, and maximize the remaining energy in the battery by the deadline. This is the \emph{energy maximization problem} \cite{elza2} most suitable for energy efficient, green applications. Finally, for the joint energy and data arrival scenario we also find the minimum delivery time of all the data packets. This is called the \emph{transmission completion time (TCT) minimization problem} \cite{Yang2012}, is important for delay limited applications. For each of these problems we identify the structure of the optimal transmission policy by solving a convex optimization problem, and based on this structure we provide an algorithm which finds the optimal transmission policy.

We next consider a more realistic model assuming only the causal knowledge of energy/data arrivals and channel gains, and study the online optimization problem. Since the optimal solution of the online optimization problem based on dynamic programming is prohibitively complex, we propose simple algorithms for the throughput and energy maximization problems based on the insights gained from the optimal solutions of the corresponding offline optimization problems.

In recent years, optimal transmission policies for EH communication systems have been studied extensively under various assumptions regarding the knowledge at the transmitter about the energy harvesting process. Within the offline optimization framework optimal transmission policies have been investigated for point-to-point \cite{Yang2012}-\cite{fade} and various multi-user communication scenarios, including broadcast channel \cite{broad}, \cite{broad2}, \cite{deniz2}, interference channel \cite{yener2} and two-hop networks \cite{deniz}, \cite{oner}. In addition, battery imperfections in terms of leakage, finite energy storage capacity, and energy storing and retrieving losses are investigated in \cite{deniz2}, \cite{finite}, and \cite{kaya}, respectively.

Online optimization of EH communication systems has also received considerable interest. Optimal transmission policies for EH nodes based on Markov decision processes are studied \cite{Jing}-\cite{Kashef}. In \cite{fade,kaya,process}, heuristic online policies are presented. A more practically oriented learning-theoretic approach to EH system optimization is studied in \cite{Blasco-Gunduz-Dohler}. See \cite{Gunduz-ComMag} for a general overview of EH communication systems under offline, online and learning-theoretic frameworks.

The effect of processing cost on EH communication systems have been investigated in \cite{process}-\cite{elza}. Optimal transmission policies that maximize the average throughput are studied for a constant single-link in \cite{process}-\cite{nossek}, and parallel channels in \cite{process}. In \cite{Gregori}, the throughput maximization problem is studied for a time-slotted system using suboptimal slot selection and power allocation. Our previous work \cite{elza} and \cite{elza2} consider a narrowband fading channel with processing cost and study the throughput maximization and energy maximization problems. The current paper extends all the prior literature by considering a broadband fading EH communication system with processing cost.

In the next section, we describe the system model. In Section \ref{preliminaries}, we summarize the glue-pouring algorithm which provides the optimal power allocation strategy in a battery operated communication system when the processing energy cost is taken into account \cite{glue}. We investigate the structure of the optimal offline transmission policies and provide directional glue-pouring interpretations for the throughput maximization, energy maximization and the TCT minimization problems in Section \ref{throughput maximization}, \ref{energy maximization}, and \ref{time minimization}, respectively. In Section \ref{online policies}, we propose online algorithms for the throughput and energy maximization problems. In Section \ref{numerical result}, numerical results are presented. Finally, we conclude our paper in Section \ref{conclude}.

\section{System Model}
\label{system}
We consider an EH transmitter communicating over a broadband fading channel modelled as $K$ parallel independently fading sub-channels. Each sub-channel has additive white Gaussian noise (AWGN) with unit variance. The real valued channel gain for sub-channel $k$ at time $t$ is denoted as $\gamma_k(t)$, $k=1,...,K$. Without loss of generality, Shannon capacity, defined as $g(p_k(t))\triangleq \frac{1}{2}\log\left(1+ \gamma_k (t)p_k(t)\right)$ (nats/s/Hz), $k=1,...,K$, is considered as the transmission rate-power function, where $p_k(t)$ is the transmission power of sub-channel $k$ at time $t$.

We assume that finite number of energy and data packets arrive at the transmitter in time interval $[0,T)$ each carrying finite amount of energy and data, respectively. We assume that the energy and data packet arrival times are denoted as $t_0^e=0<t_1^e<t_2^e<\cdots <T$ and $0\leq t_1^b<t_2^b<\cdots <t_n^b<T$, respectively. A rechargeable battery with a finite capacity of $E_{max}$ is available at the transmitter. We assume that the harvested energy is first stored in the battery before being used by the transmitter. Accordingly, the size of an harvested energy packet is less than $E_{max}$ without loss of generality. In addition, we assume that the battery is able to store and preserve the harvested energy without any loss. We also assume that $\gamma_k(t)$ changes at the time instances $0<t_{1,k}^f<t_{2,k}^f<\cdots <T$, and remains constant in between. In order to simplify the problem formulation, all channel changes and energy/data arrival events are combined in a single time series as $t_1=0<t_2<t_3<\cdots <t_I < T$ by allowing zero energy/data arrivals when the channel gain of any sub-channel changes, or the channel gains to remain constant when an energy/data packet arrives. We define an \emph{epoch} as the time interval between two consecutive events. We denote the duration of the $i$'th epoch as $\tau_i\triangleq t_{i+1}-t_{i}$. The size of the energy and data packet arriving at time $t_i$ is referred to as $E_i$ and $B_i$, respectively, and $\gamma_{i,k}$ indicates the channel gain of sub-channel $k$ in epoch $i$.

In addition to the energy used for transmission, we consider the processing energy cost of the transmitter circuitry which models the energy dissipated by the microprocessors, mixers, filters, and converters. Using the system level power consumption model of a wireless transmitter in \cite{Li}, we take into account the dependence of the processing cost on the transmission bandwidth. We assume a processing cost of $\epsilon$ joules per second for a sub-channel simplicity. This constant processing energy per sub-channel, independent of the transmission power, is consumed only during the time the corresponding sub-channel is used.

Using optimality of constant power transmission within each epoch \cite{Yang2012}, we denote the non-negative transmission power within epoch $i$ of sub-channel $k$ as $p_{i,k}$. As argued in \cite{glue}, due to the processing cost it may not be optimal to transmit continuously, i.e., bursty transmission can be optimal. Therefore, we denote transmission duration of $p_{i,k}$ as $\Theta_{i,k}$, $0 \leq \Theta_{i,k} \leq \tau_i$. Accordingly, a \textit{transmission policy} refers to power levels $p_{i,k}$ with durations $\Theta_{i,k}$, $\forall k, i$, that determine the energy allocated to each sub-channel $k$ at each epoch $i$. Any feasible transmission policy should satisfy the energy causality constraint:
\begin{eqnarray}
\label{const 1}
\sum_{j=1}^{i}{\sum_{k=1}^{K}{\Theta_{j,k} \left(p_{j,k}+\epsilon\right)}}\leq \sum_{j=1}^{i}{E_{j}}, \quad i=1,...,I.
\end{eqnarray}
Moreover, since increasing the transmission power or duration strictly increases the amount of transmitted data, an optimal transmission policy must avoid battery overflows by utilizing all the harvested energy. Therefore, an optimal transmission policy must also satisfy the following battery overflow constraint:
\begin{eqnarray}
\label{const 2}
\sum_{j=1}^{i+1}{E_{j}}-\sum_{j=1}^{i}{\sum_{k=1}^{K}{\Theta_{j,k} \left(p_{j,k}+\epsilon\right)}}\leq E_{max},  \quad i=1,...,I.
\end{eqnarray}

Data arrivals over time also impose data causality constraints on the feasible transmission policy as follows:
\begin{eqnarray}
\label{const 3}
\sum_{j=1}^{i}{\sum_{k=1}^{K}{\frac{\Theta_{j,k}}{2}\log\left(1+\gamma_{j,k}p_{j,k}\right)}}\leq \sum_{j=1}^{i}{B_{j}}, \quad i=1,...,I.
\end{eqnarray}

In Sections \ref{throughput maximization}-\ref{time minimization}, we identify the optimal \emph{offline transmission policies}, in which all energy/data arrivals and channel gains are known before transmission starts, for three different system objectives stated below. Mathematical formulations are deferred to Sections \ref{throughput maximization}, \ref{energy maximization}, and \ref{time minimization}.
\begin{itemize}
\item \textit{Throughput maximization:} Assuming that the transmitter has sufficient data in its data buffer before transmission starts, i.e., backlogged system with $B_1=\infty$, $B_i=0$, $i=2,...,I$, we maximize the total amount of data delivered by the deadline $T$.
\item \textit{Energy maximization:} Relaxing the battery size constraint, i.e., $E_{max}\rightarrow \infty$, we maximize the remaining energy in the battery by the deadline $T$ while guaranteeing that all the arriving data is delivered to the destination.
\item \textit{TCT minimization:} We minimize the delivery time of all data packets arriving at the transmitter while assuming an infinite size battery, i.e., $E_{max}\rightarrow \infty$.
\end{itemize}
In addition, we consider \emph{online transmission policies} in which we assume that all energy/data arrivals and channel gains are known causally for throughput and energy maximization problems in Section \ref{online policies}.

\section{Preliminaries}\label{preliminaries}
For ease of exposure, we first illustrate the optimal transmission policy for throughput maximization for $I=1$, $K=1$. This models a battery operated system. For a single energy arrival $E_1$ at time $t_1=0$, a channel state $\gamma$ and processing cost $\epsilon$, for $T\rightarrow \infty$, maximum throughput is given by the solution of the following optimization problem:
\begin{eqnarray}\label{prob 11}
\underset{\Theta, p: \Theta(p+\epsilon) \leq E_1}{\operatorname{max}}\; \; \frac{\Theta}{2}\log(1+\gamma p),
\end{eqnarray}
where $\Theta$ is the total transmission duration and $p$ is the transmission power. The corresponding optimal transmission power $p^*$ \cite{glue} satisfies
\begin{eqnarray}\label{eq 3}
\frac{1}{\frac{1}{\gamma}+p^*}=\frac{1}{\epsilon + p^*}{\log(1+\gamma p^*)}.
\end{eqnarray}
The above equation has only one solution for the optimal power level $p^*$ which is given by (11) in \cite{glue}. Note that $p^*$ increases as the channel gain $\gamma$ decreases\footnote[2]{This follows from (\ref{eq 3}) by taking the derivative of $p^*$ with respect to $\gamma$.}. Moreover, $p^*$ does not depend on the available energy $E_1$. For finite transmission deadline $T$, if $T \geq \frac{E_1}{p^*+\epsilon}$, then the above solution is still optimal. On the other hand, if $T<\frac{E_1}{p^*+\epsilon}$, transmitting at power $p^*$ cannot be optimal because some energy would remain in the battery at time $T$. In this case, we can increase the throughput by increasing the transmission power so that all the available energy is consumed by time $T$, and the optimal transmission power is given by $\frac{E_1}{T}-\epsilon$.\footnote[3]{A similar observation is made in \cite{deniz2} where constant rate battery leakage is considered instead of processing cost. This correspondence does not extend to multiple energy packets or fading channels as will be seen later in the paper.}

In the case of multiple fading levels, again for single sub-channel $K=1$, single energy arrival $E_1$ and no transmission deadline ($T \rightarrow \infty$), the optimal transmission policy is given by the \emph{glue-pouring algorithm} \cite{glue}. For two fading levels $\gamma_{1}>\gamma_{2}$ with durations $\tau_1$, $\tau_2$, respectively, the glue-pouring solution is summarized below. In the following, $\Theta_1$ and $\Theta_2$ denote the transmission durations for epochs with fading levels $\gamma_{1}$ and $\gamma_{2}$, and $p_1^*$ and $p_2^*$ denote the solutions of (\ref{eq 3}) for channel gains $\gamma_{1}$ and $\gamma_{2}$, respectively.
\begin{itemize}
\item If $E_1\leq \tau_1(p_1^*+\epsilon)$, then the optimal transmission policy is $\Theta_1=\frac{E_1}{p_1^*+\epsilon}$ and $\Theta_2=0$ with power levels $p_1^*$ and $0$, respectively.
\item If $\tau_1(p_1^*+\epsilon)<E_1\leq \tau_1(p_2^*+\frac{1}{\gamma_{2}}-\frac{1}{\gamma_{1}}+\epsilon)$, then the optimal transmission policy is $\Theta_1=\tau_1$ and $\Theta_2=0$ with power levels $\frac{E_1}{\tau_1}-\epsilon$ and $0$, respectively.
\item If $\tau_1(p_2^*+\frac{1}{\gamma_{2}}-\frac{1}{\gamma_{1}}+\epsilon)<E_1 \leq \tau_1(p_2^*+\frac{1}{\gamma_{2}}-\frac{1}{\gamma_{1}}+\epsilon)+\tau_2(p_2^*+\epsilon)$, then the optimal transmission policy is $\Theta_1=\tau_1$ and $\Theta_2=\frac{E_1-\tau_1(p_2^*+\frac{1}{\gamma_{2}}-\frac{1}{\gamma_{1}}+\epsilon)}{p_2^*+\epsilon}$ with power levels $p_2^*+\frac{1}{\gamma_{2}}-\frac{1}{\gamma_{1}}$ and $p_2^*$, respectively.
\item If $\tau_1(p_2^*+\frac{1}{\gamma_{2}}-\frac{1}{\gamma_{1}}+\epsilon)+\tau_2(p_2^*+\epsilon)<E_1$, then the optimal transmission policy is obtained through the classical waterfilling algorithm.
\end{itemize}

Based on the above solution, for the general system model with $K$ sub-channels \emph{glue level} in epoch $i$ of sub-channel $k$ is defined as the sum of the transmission power and the inverse channel gain in that epoch, i.e., $\frac{1}{\gamma_{i,k}}+p_{i,k}$.

\section{Throughput Maximization}
\label{throughput maximization}
In this section, we consider the throughput maximization problem introduced in Section \ref{system}, that is, we maximize the total delivered data until the deadline $T$. We assume that $B_1= \infty$ and $B_i=0$, $i=2,...,I$, and the last event corresponds to the transmission deadline, i.e., $t_{I+1}=T$. Mathematically, the problem can be formulated as follows.
\begin{subequations}
\label{prob 2c}
\begin{eqnarray}\label{prob 2c:1}
\underset{\alpha_{i,k},\Theta_{i,k}}{\operatorname{max}} \hspace{-0.2in} && \sum_{i=1}^{I}{\sum_{k=1}^{K}{\frac{\Theta_{i,k}}{2} \log\left(1+\gamma_{i,k}\frac{\alpha_{i,k}}{\Theta_{i,k}}\right)}} \\\label{prob 2c:2}
\text{s.t.} && \sum_{j=1}^{i}{\sum_{k=1}^{K}{\left(\alpha_{j,k}+\Theta_{j,k}\epsilon\right)}}-\sum_{j=1}^{i}{E_{j}} \leq 0, \quad \forall i, \\\label{prob 2c:3}
&&\sum_{j=1}^{i+1}{E_{j}}-\sum_{j=1}^{i}{\sum_{k=1}^{K}{\left(\alpha_{j,k}+\Theta_{j,k}\epsilon\right)}}\leq E_{max}, \forall i, \\ \label{prob 2c:5}
&& 0\leq \Theta_{i,k} \leq \tau_i, \quad \text{and} \quad 0\leq \alpha_{i,k},\quad \forall i, ~\forall k,
\end{eqnarray}
\end{subequations}
where we have defined $\alpha_{i,k}\triangleq\Theta_{i,k}p_{i,k}$, for $i=1,...,I$ and $k=1,...,K$. Notice that $\alpha_{i,k}$ is equivalent to the total allocated transmission energy to epoch $i$ of sub-channel $k$. In the above optimization problem, the constraints in (\ref{prob 2c:2}) and (\ref{prob 2c:3}) are due to the energy causality and battery overflow constraints in (\ref{const 1}) and (\ref{const 2}), respectively. The term $\frac{\Theta_{i,k}}{2}\log\left(1+\gamma_{i,k}\frac{\alpha_{i,k}}{\Theta_{i,k}}\right)$ is the perspective function of the concave function $\frac{1}{2}\log\left(1+\gamma_{i,k}\alpha_{i,k}\right)$. Here, we take $\frac{\Theta_{i,k}}{2}\log\left(1+\gamma_{i,k}\frac{\alpha_{i,k}}{\Theta_{i,k}}\right)=0$ when $\Theta_{i,k}=0$. Since perspective operation preserves concavity \cite{Boyd}, the objective function in (\ref{prob 2c:1}) is concave. In addition, the constraints in (\ref{prob 2c:2})-(\ref{prob 2c:5}) are linear. Therefore, the optimization problem in (\ref{prob 2c}) is convex, and efficient numerical solutions exists \cite{Boyd}.

The optimal allocated transmission energy $\alpha_{i,k}^*$ to epoch $i$ of sub-channel $k$, and the corresponding optimal transmission duration $\Theta_{i,k}^*$, for $i=1,...,I$ and $k=1,...,K$, must satisfy the following KKT conditions:
\begin{eqnarray}
\label{der 1c}
\frac{\partial \mathcal{L}}{\partial \alpha_{i,k}}&\hspace{-0.15in}=&\hspace{-0.15in}\frac{\Theta_{i,k}^* \gamma_{i,k}}{2(\Theta_{i,k}^* +\gamma_{i,k} \alpha_{i,k}^*)}- \sum_{j=i}^{I}{(\lambda_j-\mu_j)} + \sigma_{i,k}=0, \\
\label{der 2c}
\frac{\partial \mathcal{L}}{\partial \Theta_{i,k}}&\hspace{-0.15in}=&\hspace{-0.15in}\frac{1}{2}\log\left(1+\frac{\gamma_{i,k} \alpha_{i,k}^*}{\Theta_{i,k}^*}\right)- \frac{\gamma_{i,k} \alpha_{i,k}^*}{2(\Theta_{i,k}^* + \gamma_{i,k} \alpha_{i,k}^*)}-\epsilon \sum_{j=i}^{I}{(\lambda_j-\mu_j)}-\phi_{i,k}+ \psi_{i,k}=0,
\end{eqnarray}
for $i=1,...,I$ and $k=1,...,K$. Here $\mathcal{L}$ is the Lagrangian of (\ref{prob 2c}) with $\lambda_i\geq 0$, $\mu_i\geq 0$, $\phi_{i,k}\geq 0$, $\psi_{i,k}\geq 0$, and $\sigma_{i,k}\geq 0$ as Lagrange multipliers for constraints in (\ref{prob 2c:2})-(\ref{prob 2c:5}), respectively.
The complementary slackness conditions are
\begin{eqnarray}
\hspace{-0.2in}{\lambda_i \left(\sum_{j=1}^{i}{\sum_{k=1}^{K}{\left(\alpha_{j,k}^*+\Theta_{j,k}^*\epsilon\right)}}-\sum_{j=1}^{i}{E_{j}}\right)}&\hspace{-0.15in}=&\hspace{-0.15in} 0,  \forall i, \label{comp 1c:1} \\
\hspace{-0.2in}{\mu_i \left(\sum_{j=1}^{i+1}{E_{j}}-\sum_{j=1}^{i}{\sum_{k=1}^{K}{\left(\alpha_{j,k}^*+\Theta_{j,k}^*\epsilon\right)}}-E_{max}\right)}&\hspace{-0.15in}=&\hspace{-0.15in} 0, \forall i, \label{comp 1c:2} \\ \label{comp 1c:3}
\hspace{-0.2in}{\phi_{i,k} (\Theta_{i,k}^*-\tau_{i})} = 0, ~ {\psi_{i,k} \Theta_{i,k}^*} = 0, ~  {\sigma_{i,k} \alpha_{i,k}^*}=0, &&\hspace{-0.25in} \forall i, \forall k.
\end{eqnarray}

We next identify some properties of an optimal transmission policy for the throughput maximization problem based on the KKT conditions in (\ref{der 1c})-(\ref{comp 1c:3}) which are both necessary and sufficient due to the convexity of the optimization problem in (\ref{prob 2c}):

\begin{itemize}
\item If $\Theta_{i,k}^*=0$ or $\alpha_{i,k}^*=0$, then the optimal transmission power $p_{i,k}^*$ must be zero.
\item If $0<\Theta_{i,k}^*\leq \tau_i$ and $\alpha_{i,k}^* >0$, then $\psi_{i,k}= \sigma_{i,k}=0$ due to the complementary slackness conditions in (\ref{comp 1c:3}). Therefore we can compute the optimal transmission power in terms of $\lambda_i$ and $\mu_i$ as follows:
\begin{eqnarray}
\label{eq 5c}
 p_{i,k}^* =\left[\frac{1}{2\sum_{j=i}^{I}{(\lambda_j-\mu_j)}}-\frac{1}{\gamma_{i,k}}\right]^+,
\end{eqnarray}
which is obtained by substituting  $\alpha_{i,k}^*=\Theta_{i,k}^* p_{i,k}^*$ into (\ref{der 1c}). By combining (\ref{der 1c}) and (\ref{der 2c}) we can obtain
\begin{eqnarray}
\label{eq 3c}
\log\left(1+\frac{\gamma_{i,k} \alpha_{i,k}^*}{\Theta_{i,k}^*}\right)=\frac{\gamma_{i,k}(\alpha_{i,k}^* + \epsilon \Theta_{i,k}^*)}{\Theta_{i,k}^* +\gamma_{i,k} \alpha_{i,k}^*}+2\phi_{i,k}.
\end{eqnarray}
When we replace $\alpha_{i,k}^*$ in (\ref{eq 3c}) with $\Theta_{i,k}^* p_{i,k}^*$, we obtain
\begin{eqnarray}
\label{eq 2c}
\log\left(1+\gamma_{i,k} p_{i,k}^*\right)= \frac{p_{i,k}^*+\epsilon}{\frac{1}{\gamma_{i,k}}+p_{i,k}^*}+2\phi_{i,k}.
\end{eqnarray}

Note that when $0<\Theta_{i,k}^* < \tau_i$, i.e., $\phi_{i,k}=0$, (\ref{eq 2c}) is equivalent to (\ref{eq 3}). Therefore, it has a unique solution for given $\gamma_{i,k}$ and $\epsilon$. We denote the solution of (\ref{eq 2c}) when $0<\Theta_{i,k}^* < \tau_i$ as $p_{i,k}^*=v_{i,k}^*$. Since (\ref{eq 2c}) depends only on $\gamma_{i,k}$ and $\epsilon$, we can compute the optimal transmission power directly without solving the optimization problem in (\ref{prob 2c}).
When $\Theta_{i,k}^* = \tau_i$, i.e., $\phi_{i,k}\geq 0$, it can be argued from (\ref{eq 2c}) that the optimal transmission power $p_{i,k}^*$ must satisfy $p_{i,k}^*\geq v_{i,k}^*$.
\end{itemize}

\begin{rem}
\label{remark 2 rev}
When there is no processing cost, i.e., $\epsilon=0$ and $\alpha_{i,k}^*>0$, $\Theta_{i,k}^*=\tau_i$, and when $\alpha_{i,k}^*=0$, $\Theta_{i,k}^*=0$. To see this suppose $0< \Theta_{i,k}^* < \tau_i$ and $\alpha_{i,k}^*>0$. In this case we can argue that (\ref{eq 2c}) leads to $p_{i,k}^*=v_{i,k}^*=0$ when $\epsilon=0$. However this contradicts with the assumption on $\alpha_{i,k}^*>0$, since $\alpha_{i,k}^*=\Theta_{i,k}^* p_{i,k}^*=0$. Therefore, when $\epsilon=0$, there is no bursty transmission. Consequently the optimal transmission policy for $\epsilon=0$ leads to the classical water-filling over sub-channels \cite{infotheory}.
\end{rem}

\begin{lemma}
\label{lemma 1c}
In the optimal transmission policy, whenever the glue level in sub-channel $k$, i.e., $\frac{1}{\gamma_{i,k}}+p_{i,k}$, decreases (increases) from one epoch to the next, the battery must be full (empty).
\end{lemma}
\begin{proof}
The optimal transmission power satisfies (\ref{eq 5c}) whenever a non-zero transmission energy is allocated to epoch $i$ of sub-channel $k$, $i\in\{1,...,I\}$ and $k\in\{1,...,K\}$. In addition, from the complementary slackness conditions (\ref{comp 1c:1})-(\ref{comp 1c:2}), we can argue that the battery is empty whenever $\lambda_i>0$ and $\mu_i=0$, and the battery is full whenever $\lambda_i=0$ and $\mu_i>0$. This is because whenever the constraint in (\ref{prob 2c:2}) is satisfied with equality, i.e., $\lambda>0$, the constraint in (\ref{prob 2c:3}) cannot be satisfied with equality, i.e., $\mu=0$, and vice versa. From (\ref{eq 5c}) we see that $\frac{1}{\gamma_{i,k}}+p_{i,k}>\frac{1}{\gamma_{i+1,k}}+p_{i+1,k}$ implies $\lambda_i=0$ and $\mu_i>0$, since $\lambda_i=0$ and $\mu_i>0$ leads to an increase in the denominator of RHS of (\ref{eq 5c}). Similarly, $\frac{1}{\gamma_{i,k}}+p_{i,k}<\frac{1}{\gamma_{i+1,k}}+p_{i+1,k}$ implies $\lambda_i>0$ and $\mu_i=0$. Therefore, we can conclude that whenever the glue level in sub-channel $k$, $k\in\{1,...,K\}$, decreases (increases) from one epoch to the next, the battery must be full (empty).
\end{proof}

\begin{lemma}
\label{lemma 1d}
In the optimal transmission policy, the glue levels in an epoch are the same for all sub-channels to which non-zero transmission energy is allocated.
\end{lemma}

\begin{proof}
Rearranging (\ref{eq 5c}) we obtain
\begin{eqnarray}
\label{eq 6c}
\frac{1}{\gamma_{i,k}}+p_{i,k}^*=\frac{1}{2\sum_{j=i}^{I}{(\lambda_j-\mu_j)}},
\end{eqnarray}
for $\forall k\in\{k:\alpha_{i,k}^*>0\}$. Note that right hand side of (\ref{eq 6c}) must be the same for all sub-channels in epoch $i$ to which non-zero transmission energy is allocated. Therefore, we can conclude that the glue level in an epoch is the same for all sub-channels with non-zero transmission energy.
\end{proof}

\begin{rem}
\label{remark 1}
It is possible to show that $v_{i,k}^*$, the solution of (\ref{eq 2c}) when $\phi_{i,k}=0$, is a decreasing function of $\gamma_{i,k}$. Since the optimal transmission power in an epoch of sub-channel $k$ must satisfy $p_{i,k}^*\geq v_{i,k}^*$, the optimal transmission policy utilizes epochs with the highest channel gain under the energy causality and battery size constraints.
\end{rem}

\begin{rem}
\label{remark 2}
The optimization problem in (\ref{prob 2c}) may have multiple solutions. Consider a sub-channel with multiple epochs having the same channel gain. In an optimal transmission policy, if these epochs are partially utilized, i.e., $0<\Theta_{i,k}<\tau_i$, then the corresponding optimal transmission power must be equal to $v_{i,k}^*$. Then, the corresponding optimal values for $\frac{\gamma_{i,k}\alpha_{i,k}^*}{\Theta_{i,k}^*}=\gamma_{i,k}v_{i,k}^*$ in (\ref{prob 2c:1}) must also be the same, therefore, we can obtain another transmission policy by transferring some of the energy between these epochs under the energy causality and battery size constraints. Similarly, if an epoch has multiple partially utilized sub-channels having the same channel gain, we can find another optimal transmission policy by transferring energy between these sub-channels.
\end{rem}

\subsection{Directional Backward Glue-Pouring Algorithm}\label{glue algorithm}
The directional backward glue-pouring algorithm, introduced in \cite{elza} for the throughput maximization problem with a single fading channel ($K=1$), is an adaptation of the glue-pouring algorithm in Section \ref{preliminaries} to the EH model, where the energy becomes available over time. Due to the energy causality constraint, harvested energy $E_i$ can only be allocated to epochs $j \geq i$. When $E_i$ energy of amount is transferred to future epochs $j>i$, the constraint (\ref{prob 2c:2}) is satisfied with strict inequality, i.e., $\lambda_i=0$. Then the glue level cannot increase as argued in Lemma \ref{lemma 1c}. Conversely, if there is a glue level increase, that is, if $\lambda_i>0$, then the constraint (\ref{prob 2c:2}) is satisfied with equality, and no energy is transferred to future epochs. In addition, due to battery size constraint, the amount of energy that can be transferred to epoch $j$ is limited by $E_{max}-E_{j}$. When the transferred energy is less than $E_{max}-E_{j}$, the battery size constraint in (\ref{prob 2c:3}) must be satisfied with strict inequality, i.e., $\mu_i=0$, and the glue level does not change as argued in Lemma \ref{lemma 1c}. Conversely, when there is a glue level decrease, that is, if $\mu_i>0$, the amount of transferred energy to the $j$'th epoch is $E_{max}-E_{j}$. Therefore, we can allocate the harvested energy to epochs, starting from the last non-zero energy packet to the first, under the energy causality and battery size constraints. Moreover, the optimal transmission power for different sub-channels of an epoch must have the same glue level while satisfying the condition $p_{i,k}^*\geq v_{i,k}^*$. These suggest that, the optimal transmission policy can be obtained through the directional backward glue-pouring algorithm over the epochs of sub-channels. Accordingly, the optimal transmission policy can be computed as in Table \ref{table:1}.
\begin{table}[ht]
\caption{Directional backward glue-pouring algorithm}
\begin{enumerate}
\item Initialization: Set glue level for epoch $j$, $\xi_j=0$, $j=1,...,I$. Also set $i=I$.
\item Allocate $E_i$ to the subchannels of epoch $i$ using the glue pouring algorithm. Compute the glue level $\xi_i =\frac{1}{\gamma_{i,k}}+p_{i,k}^*$ while satisfying the condition $p_{i,k}^*\geq v_{i,k}^*$ for each subchannel as argued in Lemma \ref{lemma 1d}. Note that Lemma \ref{lemma 1d} guarantees $\frac{1}{\gamma_{i,k}}+p_{i,k}^*$ is the same for all $k=1,...,K$.
\item Set $m=i$. If $m=I$, go to step 6.
\item If the glue level of epoch $m$ is greater than the subsequent epoch $m+1$, i.e., $\xi_m >\xi_{m+1}$, reallocate previously allocated energies to epochs $i,...,m+1$ while satisfying the glue pouring solution within each epoch, such that the transferred energy to epoch $j$, $j \in i+1,...,m+1$, is less than and equal to $E_{max}-E_j$. Note that when the transferred energy to epoch $j$ is less than $E_{max}-E_j$, the glue level of epoch $j$ is equal to the preceding epoch $j-1$, i.e., $\xi_{j-1}=\xi_{j}$ as argued in Lemma \ref{lemma 1c}.
\item If $m=I$, go to step 6. Otherwise, increase $m$ by one, and go to step 4.
\item If $i=1$, stop. Otherwise, decrease $i$ by one and go to step 2.
\end{enumerate}
\centering
\label{table:1}
\end{table}

\begin{figure}[t]
\centering
\subfigure[]{
\includegraphics[scale=0.5,trim= 15 0 0 0]{glue_pouring.eps}
\label{fig 1:subfig1}
}
\subfigure[]{
\includegraphics[scale=0.5,trim= 15 0 0 0]{glue_pouring1.eps}
\label{fig 1:subfig2}
}
\subfigure[]{
\includegraphics[scale=0.5,trim= 15 0 0 0]{glue_pouring2.eps}
\label{fig 1:subfig3}
}
\caption{Directional backward glue-pouring algorithm.}
\label{fig 1}%
\end{figure}

To illustrate the directional backward glue-pouring algorithm, consider the example in Fig. \ref{fig 1}. There are two sub-channels ($K=2$) and two fading levels ($I=2$) in each sub-channel. The inverse channel gains $\frac{1}{\gamma_{i,k}}$ are shown as heights of the solid blocks. The dashed lines above the block are used to express optimal power levels $v_{i,k}^*$ for $0<\Theta_{i,k}^*<\tau_i$, such that $v_{i,k}^*$ corresponds to the difference in height between the solid block and the dashed one. We consider two energy arrivals at the beginning of each epoch which are indicated by the downward arrows in the figure. As argued above, the algorithm first allocates power to the second epoch using the last harvested energy $E_2$, as shown in Fig. \ref{fig 1:subfig2}, then considers the first energy packet $E_1$ for the first and second epochs together. The glue levels are the same among the sub-channels for which the condition $p_{i,k}^*\geq v_{i,k}^*$ holds, as shown in Fig. \ref{fig 1:subfig2}. Note that due to the limited battery capacity, the transferable energy from the first epoch to the second is limited by $E_{max}-E_2$, which explains the glue level difference between the first and second epochs in Fig. \ref{fig 1:subfig3}.

\section{Energy Maximization}
\label{energy maximization}
In this section, we study the energy maximization problem introduced in Section \ref{system}, that is, we maximize the remaining energy in the battery by the deadline $T$ such that all the data packets $B_i$, $i=1,...,I$, are delivered. We assume that the last event corresponds to the transmission deadline, i.e., $t_{I+1}=T$, and relax the finite battery size constraint, i.e., $E_{max} \rightarrow \infty$. The optimization problem for the energy maximization can be formulated as follows:
\begin{subequations}
\label{prob 4c}
\begin{eqnarray}\label{prob 4c:1}
\underset{\beta_{i,k},\Theta_{i,k}}{\operatorname{max}} \hspace{-0.28in}&& \sum_{i=1}^{I}{\left(E_{i}-\sum_{k=1}^{K}{\frac{\Theta_{i,k}}{\gamma_{i,k}}\left(e^\frac{2\beta_{i,k}}{\Theta_{i,k}}-1\right)+\Theta_{i,k}\epsilon}\right)} \\\label{prob 4c:2}
 \text{s.t.}\hspace{-0.1in} && \sum_{j=1}^{i}{\sum_{k=1}^{K}{\frac{\Theta_{j,k}}{\gamma_{j,k}}\left(e^\frac{2\beta_{j,k}}{\Theta_{j,k}}-1\right)+\Theta_{j,k}\epsilon}}-\sum_{j=1}^{i}{E_{j}} \leq 0, \forall i,\\\label{prob 4c:3}
&&\sum_{j=1}^{i}{\sum_{k=1}^{K}{\beta_{j,k}}}-\sum_{j=1}^{i}{B_{j}} \leq 0, \quad i=1,...,I-1, \\\label{prob 4c:4}
&&\sum_{j=1}^{I}{B_{j}}-\sum_{j=1}^{I}{\sum_{k=1}^{K}{\beta_{j,k}}} \leq 0, \\\label{prob 4c:5}
&& 0\leq \Theta_{i,k} \leq \tau_i, ~~~ \text{ and } ~~~  0\leq \beta_{i,k}, \quad \forall i, ~~\forall k,
\end{eqnarray}
\end{subequations}
where we have defined $\beta_{i,k}\triangleq \frac{\Theta_{i,k}}{2}\log\left(1+\gamma_{i,k}p_{i,k}\right)$ for $i=1,...,I$ and $k=1,...,K$. Here, $\beta_{i,k}$ can be considered as the total amount of data transmitted within epoch $i$ of sub-channel $k$. In the above optimization problem, (\ref{prob 4c:2}) and (\ref{prob 4c:3}) are due to the energy and data causality constraints in (\ref{const 1}) and (\ref{const 3}), respectively. Constraint (\ref{prob 4c:4}) arises as a result of the delivery requirement of all data packets by the deadline. Note that, the term $\Theta_{i,k}e^{\frac{2\beta_{i,k}}{\Theta_{i,k}}}$ is the perspective function of a strictly convex function $f(\beta_{i,k})=e^{2\beta_{i,k}}$. Here, we take $\Theta_{i,k}e^{\frac{2\beta_{i,k}}{\Theta_{i,k}}}=0$ when $\Theta_{i,k}=0$. Since the perspective operation preserves convexity \cite{Boyd}, the objective function in (\ref{prob 4c:1}) is concave, and the constraint set defined by (\ref{prob 4c:2})-(\ref{prob 4c:5}) is convex. Therefore, the optimization problem in (\ref{prob 4c}) is convex. The constraint set of (\ref{prob 4c}) can be empty due to insufficient harvested energy to deliver all the data packets. Feasibility of (\ref{prob 4c}) can be checked by solving the optimization problem in (\ref{prob 4c}) with a new objective function $-B$, and a new constraint $\sum_{j=1}^{I}{B_{j}}-\sum_{j=1}^{I}{\sum_{k=1}^{K}{\beta_{j,k}}}  \leq B$ replaced with (\ref{prob 4c:4}). Note that $B$ corresponds to the additional amount of data that can be delivered in the last epoch for the given energy profile. If the optimal value of this optimization problem is non-negative, i.e., $B \geq 0$, then the constraint set defined by (\ref{prob 4c:2})-(\ref{prob 4c:5}) has a feasible solution.

The optimal value of the total transmitted data $\beta_{i,k}^*$ and the corresponding transmission duration $\Theta_{i,k}^*$ for epoch $i$ of sub-channel $k$, $i=1,...,I$ and $k=1,...,K$, must satisfy the following KKT conditions:
\begin{eqnarray}\label{der 3c}
\frac{\partial \mathcal{L}}{\partial \beta_{i,k}}&\hspace{-0.15in} =&\hspace{-0.15in} \frac{2}{\gamma_{i,k}}e^{\frac{2\beta_{i,k}^*}{\Theta_{i,k}^*}}\left(1 +\sum_{j=i}^{I}{\lambda_j}\right) + \sum_{j=i}^{I-1}{\mu_j} - \mu_I-\sigma_{i,k}=0,\\\label{der 4c}
\frac{\partial \mathcal{L}}{\partial \Theta_{i,k}} &\hspace{-0.15in} =& \hspace{-0.15in} \left(\frac{2\beta_{i,k}^* e^{\frac{2\beta_{i,k}^*}{\Theta_{i,k}^*}}}{\gamma_{i,k} \Theta_{i,k}^*}-\frac{e^\frac{2\beta_{i,k}^*}{\Theta_{i,k}^*}-1}{\gamma_{i,k}}-\epsilon \right) \left(1 +\sum_{j=i}^{I}{\lambda_j}\right)-\phi_{i,k} +\psi_{i,k}=0,
\end{eqnarray}
for $i=1,...,I$ and $k=1,...,K$. Here, $\mathcal{L}$ is the Lagrangian of (\ref{prob 4c}) with $\lambda_i\geq 0$, $\mu_i\geq 0$, $\phi_{i,k}\geq 0$, $\psi_{i,k}\geq 0$, and $\sigma_{i,k}\geq 0$ as Lagrange multipliers corresponding to constraints (\ref{prob 4c:2})-(\ref{prob 4c:5}), respectively. The complementary slackness conditions are given as:
\begin{eqnarray}\label{comp 2c:1}
\lambda_i \left(\sum_{j=1}^{i}{\sum_{k=1}^{K}{\frac{\Theta_{j,k}^*}{\gamma_{j,k}}\left(e^\frac{2\beta_{j,k}^*}{\Theta_{j,k}^*}-1\right)+\Theta_{j,k}^*\epsilon}}-\sum_{j=1}^{i}{E_{j}}\right)&
\hspace{-0.2in}=&\hspace{-0.15in}0, ~ \forall i \\\label{comp 2c:2}
\mu_i \left(\sum_{j=1}^{i}{\sum_{k=1}^{K}{\beta_{j,k}^*}}-\sum_{j=1}^{i}{B_{j-1}} \right)=0,~ i=1,...,I-1&&\hspace{-0.2in}\\\label{comp 2c:3}
\hspace{-3in}\mu_I \left(\sum_{j=1}^{I}{B_{j}}-\sum_{j=1}^{I}{\sum_{k=1}^{K}{\beta_{j,k}^*}} \right)=0 \hspace{1.05in}&&\hspace{-0.15in} \\\label{comp 2c:4}
\phi_{i,k} (\Theta_{i,k}^*-\tau_i)=0, ~ \psi_{i,k} \Theta_{i,k}^* =0, ~  \sigma_{i,k} \beta_{i,k}^* =0, ~\forall i, &\hspace{-0.1in} \forall k. &\hspace{-0.2in}
\end{eqnarray}

Similar to Section \ref{throughput maximization}, we characterize the properties of the optimal transmission policy for the energy maximization problem using the KKT conditions in (\ref{der 3c})-(\ref{comp 2c:4}).

We observe that the optimal power $p_{i,k}^*$ and transmission duration $\Theta_{i,k}^*$ for epoch $i$ of sub-channel $k$ for $i=1,...,I$ and $k=1,...,K$, satisfy the following:
\begin{itemize}
\item If $\Theta_{i,k}^*=0$, $p_{i,k}^*$ must be zero as no data is transmitted in that epoch.
\item If $0<\Theta_{i,k}^*\leq \tau_i$, then  $\psi_{i,k}=\sigma_{i,k}= 0$ due to the complementary slackness conditions in (\ref{comp 2c:4}). In this case, the optimal transmission power $p_{i,k}^*$ can be computed in terms of $\lambda_j$ and $\mu_j$, $j\geq i$, as follows
\begin{eqnarray}
\label{eq 9c}
p_{i,k}^*=\left[\frac{\mu_I-\sum_{j=i}^{I-1}{\mu_j}}{2(1+\sum_{j=i}^{I}{\lambda_j})}-\frac{1}{\gamma_{i,k}}\right]^+.
\end{eqnarray}
This is obtained by using (\ref{der 3c}) and replacing $\beta_{i,k}^*$ with $\Theta_{i,k}^*\log\left(1+\gamma_{i,k}p_{i,k}^*\right)$. In addition, we can obtain the following from (\ref{der 4c}):
\begin{eqnarray}\label{eq 8c}
\frac{2\beta_{i,k}^* e^{\frac{2\beta_{i,k}^*}{\Theta_{i,k}^*}}}{\gamma_{i,k} \Theta_{i,k}^*}-\frac{e^\frac{2\beta_{i,k}^*}{\Theta_{i,k}^*}-1}{\gamma_{i,k}}-\epsilon =\frac{\phi_{i,k}}{2(1 +\sum_{j=i}^{I}{\lambda_j})}.
\end{eqnarray}
When we replace $\beta_{i,k}^*$ with $\Theta_{i,k}^*\log\left(1+\gamma_{i,k}p_{i,k}^*\right)$, we get
\begin{eqnarray}\label{eq 7c}
\log\left(1+\gamma_{i,k}p_{i,k}^*\right)\left(\frac{1}{\gamma_{i,k}}+p_{i,k}^*\right)-(p_{i,k}^*+\epsilon) =\frac{\phi_{i,k}}{2(1 +\sum_{j=i}^{I}{\lambda_j})}
\end{eqnarray}

Note that when $0<\Theta_{i,k}^*< \tau_i$, i.e., $\phi_{i,k}=0$, we obtain (\ref{eq 3}) since $\left(1 +\sum_{j=i}^{I}{\lambda_j}\right)>0$. This suggests that the optimal transmission power $p_{i,k}^*$ is equal to $v_{i,k}^*$, the solution of (\ref{eq 2c}) when $\phi_{i,k}=0$. When $\Theta_{i,k}^*=\tau_i$, i.e., $\phi_{i,k} \geq 0$, it can be argued from (\ref{eq 7c}) that the optimal transmission power $p_{i,k}^*$ must satisfy $p_{i,k}^* \geq v_{i,k}^*$.
\end{itemize}

\begin{rem}
\label{remark 3 rev}
Similar to the throughput maximization problem in Section \ref{throughput maximization}, as argued in Remark \ref{remark 2 rev}, the optimal transmission policy over sub-channels becomes the classical water-filling solution when there is no processing cost, i.e., $\epsilon=0$. This follows from the fact that (\ref{eq 7c}) leads to $p_{i,k}^*=v_{i,k}^*=0$, when $\epsilon=0$ and $0< \Theta_{i,k}^* < \tau_i$, as argued in Remark \ref{remark 2 rev}.
\end{rem}

\begin{lemma}
\label{lemma 2c}
In the optimal transmission policy, whenever the glue level in sub-channel $k$, $k\in\{1,...,K\}$, increases from one epoch to the next, either the battery depletes and a new energy packet is harvested, or the data buffer empties and a new data packet arrives.
\end{lemma}
\begin{proof}
The optimal transmission power $p_{i,k}^*$ satisfies (\ref{eq 9c}) when
there is non-zero data transmission in epoch $i$ of sub-channel $k$ for $i=1,...,I$ and $k=1,...,K$. We can also conclude from the complementary slackness conditions in (\ref{comp 2c:1})-(\ref{comp 2c:2}) that whenever $\lambda_i>0$, the battery depletes, and whenever $\mu_i>0$, the data buffer empties. Therefore, the glue level increases from one epoch to the next, when either the battery depletes and a new energy packet is harvested, or the data buffer empties and a new data packet arrives.
\end{proof}

Similar to Lemma \ref{lemma 1d}, in the optimal transmission policy, the glue levels in an epoch are the same for all the sub-channels $k\in\{k:\beta_{i,k}^*>0\}$.

Note as in Section \ref{throughput maximization}, the optimal transmission policy must satisfy $p_{i,k}^*\geq v_{i,k}^*$. Therefore, Remark \ref{remark 1} is valid for the energy maximization problem as well.

\begin{rem}
\label{remark 3}
Similar to the throughput maximization problem in (\ref{prob 2c}), the energy maximization problem in (\ref{prob 4c}) may have multiple solutions. The optimal transmission power $p_{i,k}^*$ is equal to $v_{i,k}^*$ if the optimal transmission duration of an epoch of a sub-channel satisfies $0<\Theta_{i,k}^*<\tau_i$. If multiple epochs have the same channel gain, the optimal values satisfying $\frac{\beta_{i,k}^*}{\Theta_{i,k}^*}=\frac{1}{2}\log(1+\gamma_{i,k} p_{i,k}^*)$ are the same. Therefore, as can be argued from the objective function of (\ref{prob 4c}), we can find another optimal transmission policy satisfying the energy and data causality constraints by transmitting some of the data in a different epoch with the same optimal transmission power.
\end{rem}

\subsection{Directional Backward Glue-Pouring Algorithm with Data Arrivals}
\label{alg 2}
The directional glue-pouring algorithm of Section \ref{glue algorithm} can be modified to solve the energy maximization problem by taking into account data arrivals. A data packet can only be transmitted after it has arrived due to the data causality constraint. When part of the data $B_i$ is transferred to future epochs $j>i$, the constraint (\ref{prob 4c:3}) is satisfied with inequality, i.e., $\mu_i=0$. Then the glue level remains the same as argued in Lemma \ref{lemma 2c}. Conversely, if there is a glue level increase, i.e., if $\mu_i>0$, then the constraint (\ref{prob 4c:3}) is satisfied with equality, and no data is transferred to future epochs. By Lemma \ref{lemma 1c} the optimal transmission policy must satisfy the condition $p_{i,k}^*\geq v_{i,k}^*$, and the glue levels are the same for all sub-channels in an epoch. Accordingly, we can schedule transmission of the data starting from the last non-zero data packet to the first, such that the required energy to transmit the data satisfies the energy causality constraint. Therefore, the optimal transmission policy can be computed using a directional backward glue-pouring algorithm with data arrivals in which the data packet $B_i$ is transmitted over subsequent epochs, and the energy allocation for each data packet is done using the glue-pouring algorithm in Section \ref{glue algorithm}. Accordingly, the optimal transmission policy can be computed as in Table \ref{table:2}.
\begin{table}[ht]
\caption{Directional backward glue-pouring algorithm with data arrivals}
\begin{enumerate}
\item Initialization: Set glue level for epoch $j$, $\xi_j=0$, $j=1,...,I$. Also set $i=I$.
\item Allocate energy to the subchannels of epoch $i$ using the glue pouring algorithm such that $B_i$ amount of data is delivered in that epoch. Compute the glue level $\xi_i =\frac{1}{\gamma_{i,k}}+p_{i,k}^* $ while satisfying the condition $p_{i,k}^*\geq v_{i,k}^*$ for each subchannel as argued in Lemma \ref{lemma 1d}.
\item Set $m=i$. If $m=I$, go to step 6.
\item If the glue level of epoch $m$ is greater than the subsequent epoch $m+1$, i.e., $\xi_{m} > \xi_{m+1}$, reallocate power to the subchannels of epochs $i,...,m+1$ while satisfying the glue pouring solution within each epoch, such that the allocated energy to epochs $j$, $j=i,...,n$, $n\leq m+1$, is less than and equal to $\sum_{j=1}^{n}{E_j}$. Note that when the allocated energy to epochs $j$, $j=i,...,n$ is less than $\sum_{j=1}^{n}{E_j}$, the glue level of each epoch is constant.
\item If $m=I$, go to step 6. Otherwise, increase $m$ by one, and go to 4.
\item If $i=1$, stop. Otherwise, decrease $i$ by one and go to step 2.
\end{enumerate}
\centering
\label{table:2}
\end{table}

\begin{figure}[t]
\centering
\subfigure[]{
\includegraphics[scale=0.5,trim= 15 0 0 0]{data_pouring.eps}
\label{fig 2:subfig1}
}
\subfigure[]{
\includegraphics[scale=0.5,trim= 15 0 0 0]{data_pouring1.eps}
\label{fig 2:subfig2}
}
\subfigure[]{
\includegraphics[scale=0.5,trim= 15 0 0 0]{data_pouring2.eps}
\label{fig 2:subfig3}
}
\caption{Directional backward glue-pouring algorithm with data arrivals.}
\label{fig 2}%
\end{figure}

To illustrate the directional backward glue-pouring algorithm with data arrivals, we consider the algorithm for two sub-channels with two fading states in each sub-channel as shown in Fig. \ref{fig 2}. The inverse channel gains are indicated by solid blocks in the figure. The optimal power levels $v_{i,k}^*$ are indicated with dashed lines which are $v_{i,k}^*$ above the solid blocks. In addition, the energy and data arrivals are showed with downward arrows, respectively. The algorithm first allocates power to the second epoch such that $B_2$ bits are transmitted in this epoch and the glue levels are the same in the sub-channels in which the condition $p_{2,k}^* \geq v_{2,k}^*$, $k=1,2$, is satisfied (see Fig. \ref{fig 2:subfig2}). Note that, although both energies $E_1$ and $E_2$ are available for the transmission of $B_2$ bits, $E_2$ is used first, as $E_1$ can also be used to transmit the bits in the first data packet. If $E_2$ was not enough to transmit $B_2$ bits, some of the energy from the first arrival $E_1$ would also be used. Then, the algorithm considers the first data packet $B_1$ and allocates power according to the glue-pouring algorithm in Section \ref{glue algorithm} as shown in Fig. \ref{fig 2:subfig3}.

\section{Transmission Completion Time (TCT) Minimization}
\label{time minimization}
In this section, we consider the TCT minimization problem introduced in Section \ref{system}. Our goal is to identify an optimal transmission policy which minimizes the delivery time of all the data packets $B_i$, $i=1,...,I$. We again assume that the battery has infinite size $E_{max} \rightarrow \infty$. We first discuss the relation between the TCT minimization and energy maximization problems, and then we propose an algorithm which finds the optimal transmission policy for TCT minimization.

As argued in Remark \ref{remark 3}, the optimal transmission scheme for the energy maximization problem may have multiple solutions which can lead to different TCTs. Since we are seeking the minimum TCT, without loss of optimality, we put some restrictions on the optimal transmission policy obtained by the energy maximization problem before we relate the two problems: i) non-zero power is always allocated at the beginning of an epoch, i.e., during the time interval $\left[t_{i-1},t_{i-1}+\Theta_{i,k}^*\right)$; ii) if there are multiple epochs with the same channel gain in a sub-channel $k$, $k=1,...,K$, the transmission power is allocated starting from the earliest epoch satisfying the energy and data causality constraints; and iii) if all the utilized sub-channels in the last epoch have the same channel gain, then the transmission power is allocated to those sub-channels for which the transmission duration $\Theta_{i,k}^*$ is the same.

Denoting the minimum TCT time as $T_{min}$ we note that the battery must be depleted by the time $T_{min}$, otherwise we could increase the transmission power and deliver the arrived data in a shorter time. Therefore, we can conclude that the remaining energy in the battery obtained by the energy maximization problem must be zero when the deadline $T$ is equal to $T_{min}$. Any delay constraint $T$, for which the energy maximization problem leads to zero remaining energy in the battery, satisfies $T \geq T_{min}$, as the transmission power in the time interval $[T_{min},T)$ can be zero.

Following the arguments above, the smallest transmission deadline $T=t_m$, for which the energy maximization problem has a feasible solution, is an upper bound on $T_{min}$. This suggests that for $T=t_{m-1}$, the harvested energy is insufficient to transmit all the arrived data packets, and $t_{m-1}$ is a lower bound on $T_{min}$.
Note that due to the requirement of transmitting all the arriving data packets, we also need to ensure that the last non-zero data packet arrival instant $t_n^b$ is upper bounded by $t_m$. After identifying the time interval $(t_{m-1},t_m]$, we formulate a convex optimization problem which minimizes the maximum of the transmission durations of sub-channels in epoch $m$, i.e., $\max\{\Theta_{m,k} : \forall k, T_{min} \in (t_{m-1},t_m]\}$ to find $T_{min}$. The TCT minimization algorithm is outlined next.

\subsection{TCT minimization}
 In order to compute $T_{min}$, we first find the smallest $m\in\{1,...,I\}$, such that $t_m$ is greater than the last nonzero data packet arrival time $t_n^b$, and the directional backward glue-pouring algorithm in Section \ref{alg 2} with $T=t_m$ has a feasible solution.

 We next solve the following minimization problem:
\begin{subequations}
\label{prob tt}
\begin{eqnarray}\label{prob tt:1}
\underset{\beta_{i,k},\Theta_{i,k}}{\operatorname{min}}\hspace{-0.25in}&& t \\\label{prob tt:2}
 \text{s.t.}\hspace{-0.1in} && \sum_{j=1}^{i}{\sum_{k=1}^{K}{\frac{\Theta_{j,k}}{\gamma_{j,k}}\left(e^\frac{2\beta_{j,k}}{\Theta_{j,k}}-1\right)+\Theta_{j,k}\epsilon}}-\sum_{j=1}^{i}{E_{j}} \leq 0, ~~ i=1,...,m,\\\label{prob tt:3}
&&\sum_{j=1}^{i}{\sum_{k=1}^{K}{\beta_{j,k}}}-\sum_{j=1}^{i}{B_{j}} \leq 0, \quad i=1,...,m-1, \\\label{prob tt:4}
&&\sum_{j=1}^{m}{B_{j}}-\sum_{j=1}^{m}{\sum_{k=1}^{K}{\beta_{j,k}}} \leq 0, \\\label{prob tt:5}
&& 0\leq \Theta_{i,k} \leq \tau_i, ~ i=1,...,m-1, ~ k=1,...,K, \\\label{prob tt:6}
&& 0\leq \Theta_{m,k} \leq t, \quad k=1,...,K, \\\label{prob tt:7}
&& 0\leq \beta_{i,k}, \quad i=1,...,m, ~~~ k=1,...,K,
\end{eqnarray}
\end{subequations}
where $\beta_{i,k}\triangleq \frac{\Theta_{i,k}}{2}\log(1+\gamma_{i,k}p_{i,k})$ for $i=1,...,N$ and $k=1,...,K$, and $t$ is the epigraph of $\max\{\Theta_{m,k} : k=1,...,K \}$ as stated in (\ref{prob tt:6}). Here, $\beta_{i,k}$ can be considered as the total amount of data transmitted within epoch $i$ of sub-channel $k$.  In the above optimization problem (\ref{prob tt:2}) and (\ref{prob tt:3}) are due to the energy and data causality constraints in (\ref{const 1}) and (\ref{const 3}), respectively.
The minimum TCT is $T_{min}=t_{m-1}+t^*$, where $t^*$ is the solution of the above optimization problem.
Once $T_{min}$ is found, the corresponding optimal transmission policy can be obtained by solving the energy maximization problem in Section \ref{energy maximization} with deadline $T=T_{min}$.

\section{Online Transmission Policies}\label{online policies}
In this section, we consider causal knowledge of the energy and data arrival profiles and channel gains at the transmitter. In such a scenario, the optimal online transmission policy can be obtained by first discretizing the state space and applying dynamic programming \cite{dynamic}. However, due to the high computational complexity of dynamic programming algorithms, here we focus on less complex heuristic online algorithms for the throughput and energy maximization problems using properties of the offline optimal transmission policies developed in Sections \ref{throughput maximization} and \ref{energy maximization}. Numerical comparisons with the optimal offline policies and dynamic programming solutions will be provided in Section \ref{numerical result}.

\subsection{Throughput Maximization}
\label{throughput online}
The proposed online throughput maximizing transmission policy is of myopic nature. The algorithm allocates transmission power to sub-channels based on the available energy in the battery and channel gains of sub-channels whenever an event (a variation in the channel gain or an energy arrival) occurs. Since consuming all the harvested energy by the deadline is optimal, transmission powers over the sub-channels are computed such that the battery is depleted by the deadline as if there will be no more energy arrivals or channel state variations. Therefore, available energy at the battery, which is bounded by $E_{max}$, is allocated to sub-channels using the directional backward glue pouring algorithm as introduced in Section \ref{glue algorithm}. As argued in Lemma \ref{lemma 1d}, the optimal glue level must be the same for all sub-channels to which non-zero transmission energy is allocated. The transmitter continues its transmission using the optimal transmission powers resulting from the above computation until either the battery depletes, or a new event occurs.

\subsection{Energy Maximization}
\label{energy online}
The online energy maximization problem is also a myopic one. Since the transmitter does not know the future energy/data packet arrivals or the channel gains, the online policy evaluates the transmission power $p_{i,k}$ for each sub-channel at time $t_{i}$ based on the available data in the data buffer and the channel gains of the sub-channels. The energy maximization problem requires transmitting all the data packets by the deadline $T$. Therefore, the transmission powers $p_{i,k}$, $k=1,...,K$, have to be chosen to guarantee the transmission of all stored data at time $t_{i}$ until the deadline $T$ as if there are no energy/data arrivals or channel gain changes after $t_{i}$. Accordingly, the optimal transmission powers $p_{i,k}$ can be computed using the directional backward glue pouring algorithm as introduced in Section \ref{alg 2}. Then the transmission powers are set to $p_{i,k}$ and the transmission durations to $\Theta_{i,k}$ over the respective sub-channels until either a new event occurs, or the battery depletes due to insufficient energy to transmit all the data. As argued in Section \ref{energy maximization}, the optimal glue level is the same for all the utilized sub-channels while the optimal transmission power satisfies the condition $p_{i,k}^* \geq v_{i,k}^*$.

\section{Numerical Results}\label{numerical result}
In this section, we provide numerical results to illustrate the optimization problems considered. We first study the offline throughput maximization problem. We consider four parallel sub-channels with three epochs with durations $\mathbf{\tau}=[3.5,4,2.5]$ s. We consider an energy arrival profile $\mathbf{E}=[9,8,5]$ microjoules ($\mu$J) at the beginning of each epoch. We set $E_{max}=10$ $\mu$J. Channel gains of epochs are $\gamma_{.,1}=(0.8,0.55,0.45)\times 10^6$, $\gamma_{.,2}=(0.35,0.9,0.6)\times 10^6$, $\gamma_{.,3}=(0.6,0.4,0.5)\times 10^6$, and $\gamma_{.,4}=(0.55,0.35,0.4)\times 10^6$ for sub-channels $k=1,2,3,4$, respectively. The optimal transmission policy for the throughput maximization problem for the above energy and channel profile for different values of the processing cost $\epsilon$ is shown in Fig. \ref{fig 7}. In Fig. \ref{fig 7}, all the sub-channels in an epoch are shown as a sequence of blocks which are labeled with the corresponding sub-channel index. The solid blocks represent the inverse channel gains, the dashed horizontal lines correspond to $\frac{1}{\gamma_{i,k}}+v_{i,k}^*$, where $v_{i,k}^*$ is the solution of (\ref{eq 2c}) when $\phi_{i,k}=0$ and the shaded blocks show the optimal power levels. The optimal transmission policy with no processing cost, i.e., $\epsilon=0$, is shown in Fig. \ref{fig 7:subfig1}. As can be seen from the figure, since there is no cost of increasing the transmission duration, the optimal transmission policy across sub-channels is classical water-filling (Remark \ref{remark 2 rev}). The difference in power levels among epochs is due to the energy causality constraint as argued in Section \ref{throughput maximization}. For the same energy arrival and channel profile, taking into account a processing cost of $\epsilon=0.25$ $\mu$W per sub-channel, we obtain the transmission policy in Fig. \ref{fig 7:subfig3}. The processing cost results in the total transmitted data falling from $6.23$ nats to $5.21$ nats. As shown in the figure, the optimal transmission policy becomes bursty while having the same glue level within an epoch. In the figure, the decrease in the optimal glue level from the first epoch to the second is due to the finite battery size, and the increase in the optimal glue level from the second epoch to the third is due to the energy causality constraint.
\begin{figure}[ht]
\centering
\subfigure[]{
\includegraphics[scale=0.6,trim= 45 0 20 0]{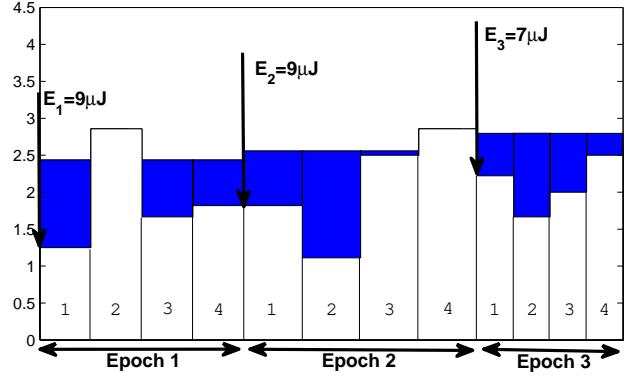}
\label{fig 7:subfig1}
}
\subfigure[]{
\includegraphics[scale=0.6,trim= 58 0 35 0]{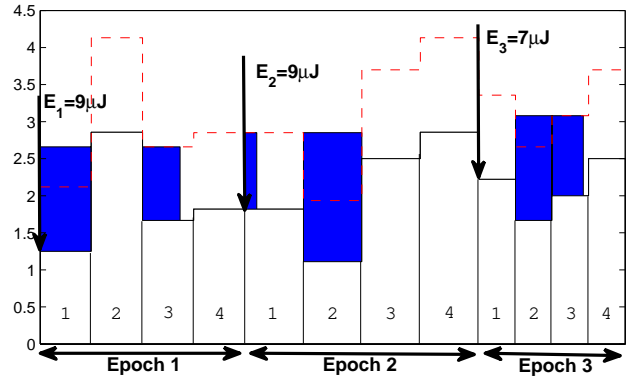}
\label{fig 7:subfig3}
}
\caption{Throughput maximization: (a) Optimal power levels for $\mathbf{\epsilon}=0$, shown as the heights of the shaded blocks, are $\left([1.18, 0.74, 0.57],\right.$ $\left.[0, 1.44, 1.13],[0.76, 0.05, 0.79],[0.61, 0, 0.29]\right)$ $\mu$W with durations $\left([3.5,4,2.5],[0,4,2.5],[3.5,4,2.5],[3.5,0,2.5]\right)$ s for sub-channels $k=1,...,4$, respectively. Total transmitted data is $B=6.23$ nats. (b) Optimal power levels for $\mathbf{\epsilon}=0.25$ $\mu$W, shown as the heights of the shaded blocks, are $\left([1.4, 1.03, 0],[0, 1.74, 1.41],[1, 0, 1.08],[0, 0, 0]\right)$ $\mu$W with durations $\left([3.5,0.8,0],[0,4,2.56],[2.04,0,2.13],[0,0,0]\right)$ s for sub-channels $k=1,...,4$, respectively. Total transmitted data is $B=5.21$ nats.}
\label{fig 7}%
\end{figure}

In Fig. \ref{fig 8} we illustrate the variation of the throughput with respect to $\epsilon$ for the same energy and channel profile given above. In addition, we illustrate the total transmission duration, which is the sum of the maximum of the transmission durations of all the sub-channels in an epoch, with respect to the processing energy cost in Fig. \ref{fig 81}. As it can be seen in the figure, as the processing energy cost increases, the transmission becomes more bursty.
\begin{figure}[ht]
\centering
\includegraphics[scale=0.7,trim= 25 0 0 0]{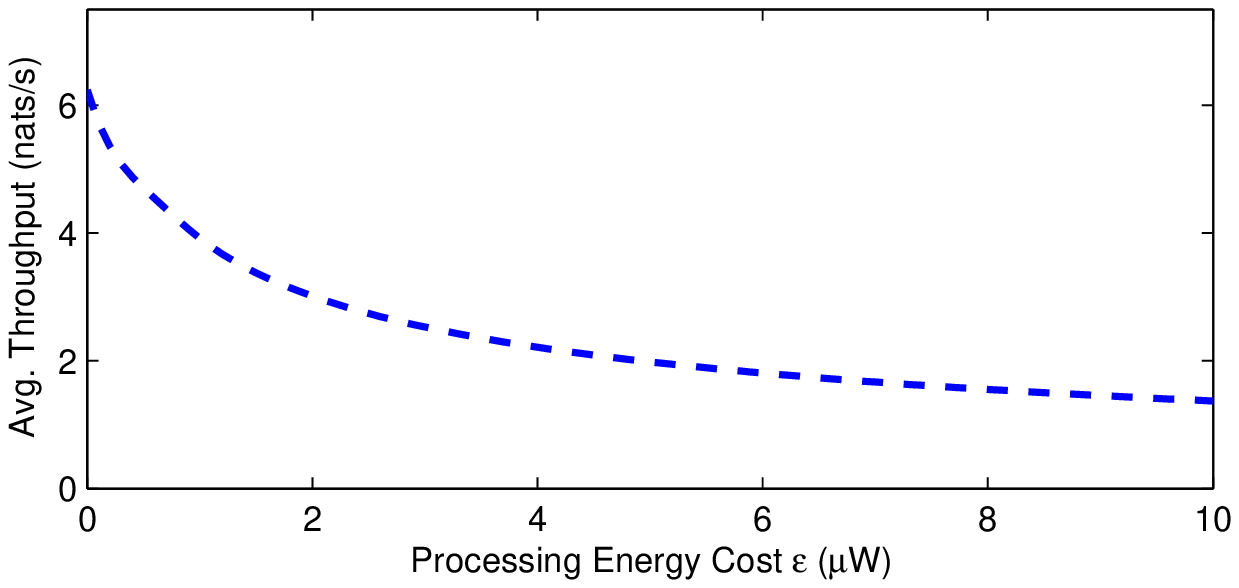}
\caption{Average throughput versus processing energy cost.}
\label{fig 8}%
\end{figure}
\begin{figure}[ht]
\centering
\includegraphics[scale=0.7,trim= 25 0 0 0]{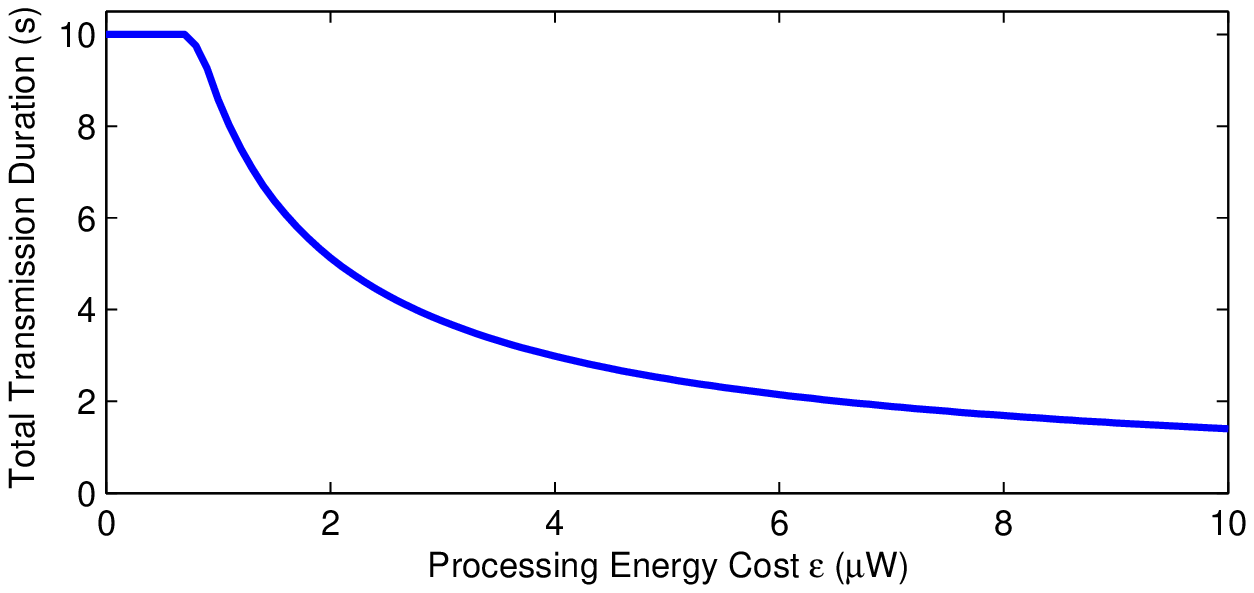}
\caption{Total transmission duration versus processing energy cost.}
\label{fig 81}%
\end{figure}

We next illustrate the optimal offline transmission policy for the energy maximization problem for different processing energy costs. We use the same energy arrival and channel gain profile given above and a data profile $\mathbf{B}=[0.5,2,1.5]$ nats. First, we set $\epsilon=0$ for each sub-channel, and obtain the transmission policy shown in Fig. \ref{fig 6:subfig1}. As shown in the figure, the optimal transmission policy utilizes epochs fully as there is no cost in increasing the transmission duration. In this case, the optimal transmission policy across sub-channels is classical water-filling (Remark \ref{remark 3 rev}), and the water level increases monotonically within a sub-channel due to the energy and data causality constraints. The remaining energy in the battery is $6.5$ $\mu$J. Then, we set $\epsilon=0.25$ $\mu$W, and obtain the optimal transmission policy in Fig. \ref{fig 6:subfig3}. As shown in the figure, the optimal transmission policy is now bursty. Consistent with the observations in Section \ref{energy maximization}, the glue levels are the same within an epoch, and increase monotonically within a sub-channel due to energy and data causality constraints. The remaining energy in the battery is $2.54$ $\mu$J.
\begin{figure}[ht]
\centering
\subfigure[]{
\includegraphics[scale=0.6,trim= 45 0 20 0]{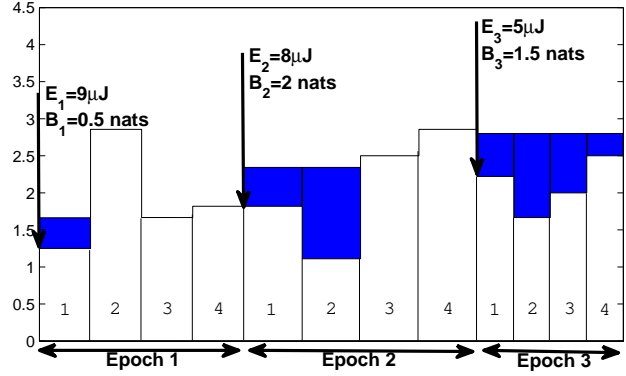}
\label{fig 6:subfig1}
}
\subfigure[]{
\includegraphics[scale=0.6,trim= 58 0 35 0]{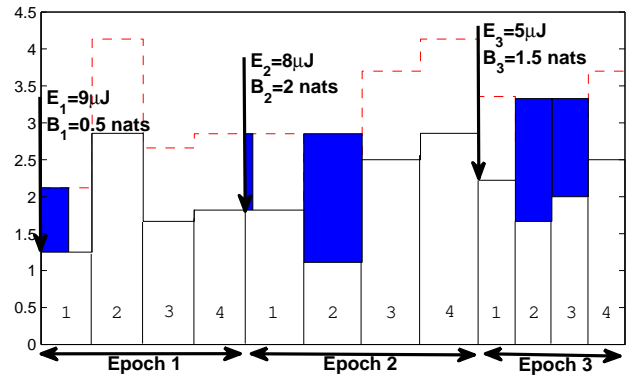}
\label{fig 6:subfig3}
}
\caption{Energy maximization: (a) Optimal power levels for $\epsilon=0$, shown as the heights of the shaded blocks, are $\left([0.41, 0.52, 0.57],[0, 1.23, 1.13],[0, 0, 0.8],[0, 0, 0.3]\right)$ $\mu$W with durations $\left([3.5,4,2.5],[0,4,2.5],[2.05,0,2.5],[0,0,2.5]\right)$ s for sub-channels $k=1,...,4$, respectively. The remaining energy in the battery is $6.5$ $\mu$J. (b) Optimal power levels for $\epsilon=0.25$ $\mu$W, shown as the heights of the shaded blocks, are $\left([0.87, 1.03, 0],[0, 1.74, 1.66],[0, 0, 1.32],[0, 0, 0]\right)$ $\mu$W with durations $\left([1.87,0.51,0],[0,4,2.5],[0,0,2.5],[0,0,0]\right)$ s. The remaining energy in the battery is $2.54$ $\mu$J.}
\label{fig 6}%
\end{figure}

The variation of the remaining energy in the battery with respect to $\epsilon$ for the above energy/data arrival and channel gain profile is shown in Fig. \ref{fig 9}. We observe that the maximum energy that can be saved in the battery at the end of the deadline decreases rapidly as the processing cost increases. For a processing cost of $\epsilon=0.49$ $\mu$W, the arriving energy is exactly the amount that is needed to transmit the arriving data. Transmission of all the data by the deadline is not possible for $\epsilon>0.49$ $\mu$W.
\begin{figure}[ht]
\centering
\includegraphics[scale=0.78,trim= 25 0 0 0]{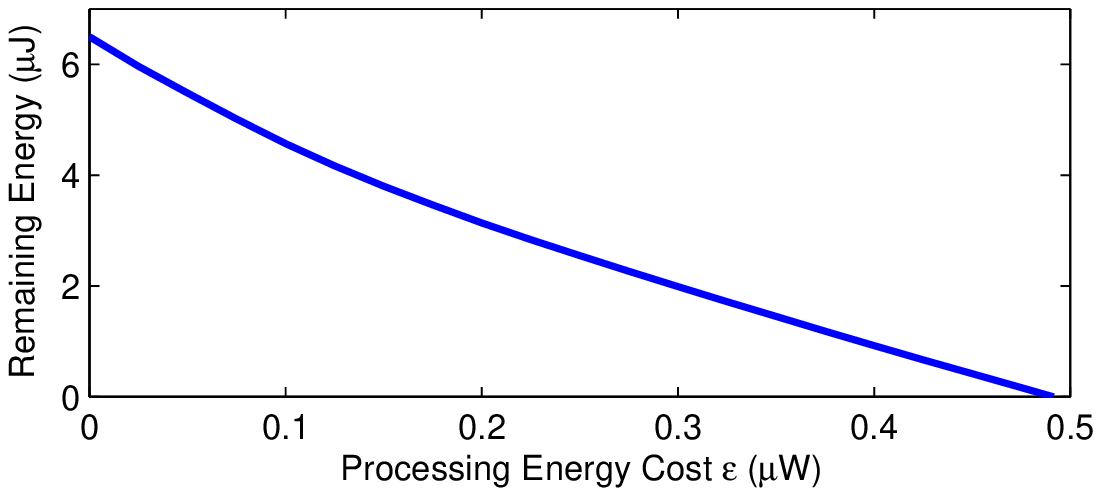}
\caption{Remaining energy in the battery versus processing energy cost.}
\label{fig 9}%
\end{figure}

We also consider the offline TCT minimization problem for the above energy/data arrival and channel gain profile with processing cost $\epsilon=0.25$ $\mu$W. The corresponding optimal transmission policy is given in Fig. \ref{fig 10}. The  corresponding minimum TCT is $8.26$ s.
\begin{figure}[ht]
\centering
\includegraphics[scale=0.6,trim= 25 0 0 0]{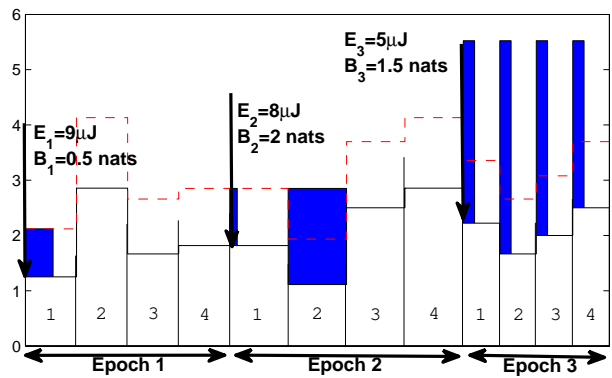}
\caption{TCT minimization: Optimal power levels for $\epsilon=0.25$, shown as the heights of the shaded blocks, are $\left([0.87, 1.03, 3.02],[0, 1.74, 3.86],[0, 0, 3.52],[0,0,3.02]\right)$ $\mu$W with durations $\left([1.89,0.51,0.76],[0,4,0.76],[0,0,0.76],[0,0,0.76]\right)$ s for sub-channels $k=1,...,4$, respectively. The minimum transmission completion time is $8.26$ s.}
\label{fig 10}%
\end{figure}

Finally, we evaluate the average performance of the online algorithms of Section \ref{online policies} by comparing them with the corresponding optimal offline policies. We consider two sub-channels. Each sub-channel has a fixed channel gain for $1$ s, which is independent across sub-channels and fading blocks, drawn from an exponential distribution with parameter $\lambda=1$. We set the transmission deadline to $T=10$ s. Therefore, there are ten fading levels for each sub-channel. We also assume that energy/data packets arrive only when the channel gains change. We first illustrate the performance of the throughput maximization problem. We set the battery size to $E_{max}=10$ $\mu$J, and the processing cost of the sub-channels to $\epsilon=1$ $\mu$W, respectively. We assume that energy packets have sizes chosen from a uniform distribution in the interval $[0, E]$ $\mu$J, where $E \in (0,10]$ $\mu$J. In order to see the degradation in the performance of the proposed online algorithm, we also provide a dynamic programming based solution \cite{dynamic}. The dynamic programming solution requires the quantization of battery state, energy amounts and fading states, and it achieves optimal performance asymptotically as the quantization resolution becomes finer. In our simulation we quantize the amount of energy in the battery uniformly with step size $1$ $\mu$J. We also quantize the fading states into eight levels such that levels are uniformly distributed. We illustrate the average throughput as a function of the average energy arrival rate $\frac{E}{2}$ in Fig. \ref{fig 11}. As shown in the figure, the online algorithm performs close to the offline transmission policy despite the lack of information about the future events. It also performs close to the dynamic programming solution. The performance loss of the online algorithm at high energy rates is partly due to the increased probability of battery overflows.

\begin{figure}[t]
\centering
\subfigure[]{
\includegraphics[scale=0.75,trim= 18 0 0 0]{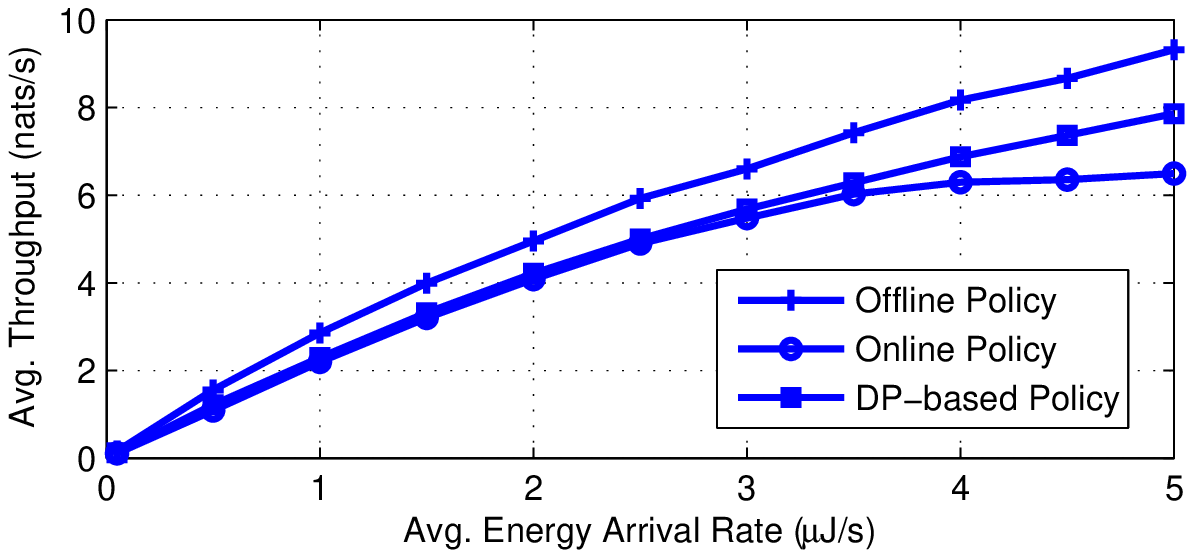}
\label{fig 11}%
}
\subfigure[]{
\includegraphics[scale=0.75,trim= 48 0 35 0]{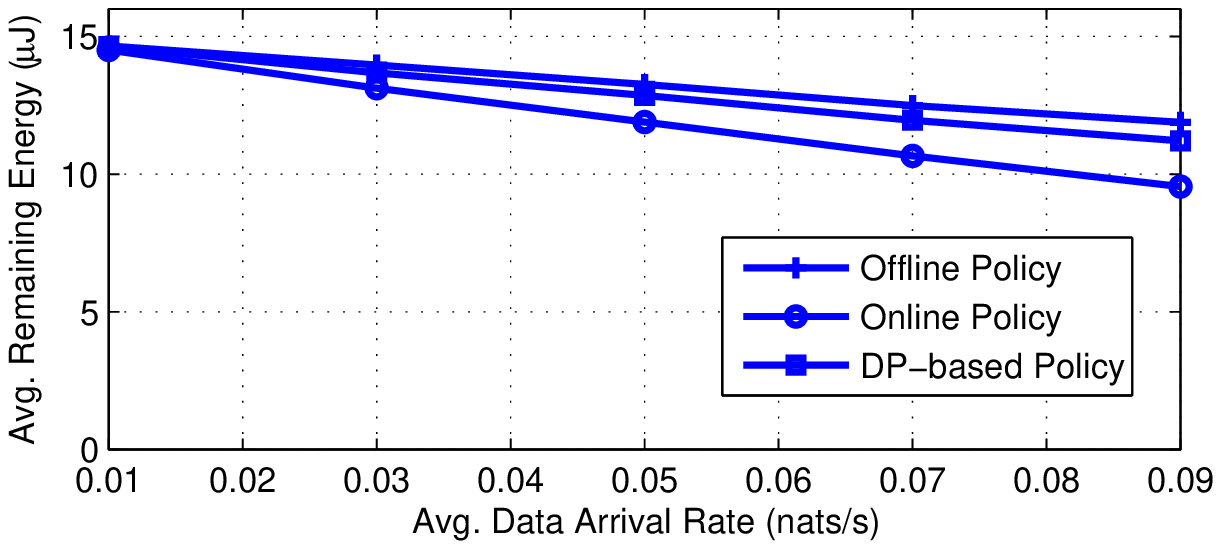}
\label{fig 13}%
}
\caption{(a) Average performances of online and offline throughput maximization policies as a function of the energy arrival rates. Dynamic programming (DP) solution is also included.  (b) Average performances of online and offline energy maximization policies as a function of the data arrival rates. Dynamic programming (DP) solution is also included.}
\end{figure}

Next, we evaluate the performance of the online energy maximization policy. We assume that the sizes of the energy and data packets are chosen from uniform distribution over the intervals $[0,3]$ $\mu$J and $[0,B]$ nats, where $B \in (0,0.18]$ nats, respectively. Similarly to the throughput maximization problem, to see the degradation in the performance of the proposed online algorithm, we also provide a dynamic programming based solution. We quantize the energy levels and fading states as in the throughput maximization problem. We also quantize the amount of data in the data buffer uniformly with step size $0.01$ nats. We demonstrate the average remaining energy in the battery as a function of the average data arrival rate $\frac{B}{2}$ in Fig. \ref{fig 13}. Some of the energy/data and channel gain realizations lead to infeasible solutions. Therefore, the remaining energy in the battery is averaged only over the feasible cases. When the average energy and data arrival rates are low, the optimal power allocation is mostly bursty, and the proposed online algorithms perform closer to the offline ones as seen in Fig. \ref{fig 11} and \ref{fig 13}. However, as the energy/data rate increases the information about the future events becomes more significant, and the lack of information on the future energy/data arrivals leads to a degradation in the performance of the online algorithm as well as the dynamic programming based policy as seen in Fig. \ref{fig 11} and \ref{fig 13}.

\section{Conclusions}
\label{conclude}
We have studied a broadband energy harvesting communication system modelled as having $K$ parallel fading channels by considering both the transmission and processing energy for each sub-channel. We have identified the optimal offline transmission policies for three different objectives; maximization of the transmitted data by a deadline, maximization of the remaining energy in the battery by a deadline and minimization of the TCT of all the arriving data packets. For the throughput and energy maximization problems we have formulated a convex optimization problem and identified properties of the optimal transmission policies. We have then discussed the relation between the energy maximization and the TCT minimization problems. We have also provided algorithms which compute the optimal transmission policies for all the three problems. Moreover, for the case the energy/data arrivals and channel gains are known causally, we have suggested myopic online algorithms for throughput and energy maximization. We have shown that the proposed low-complexity online algorithms perform close to the dynamic programming solution and the offline policies at low energy/data arrival rates. Finally, numerical results have been presented to illustrate the effect of the processing cost on the optimal transmission policies and the performance in both the offline and online settings.


\begin{thebibliography}{1}
\bibitem{bahai}
S. Cui, A. J. Goldsmith, and A. Bahai, "Energy-constrained modulation optimization," \emph{IEEE Trans. Wireless Commun.}, vol. 4, no. 5, pp. 2349-2360, Sep. 2005.
\bibitem{glue}
P. Youssef-Massaad, L. Zheng, and M. Medard, "Bursty transmission and glue pouring: On wireless channels with overhead costs," \emph{IEEE Trans. Wireless Commun.}, vol. 7, no. 12, pp. 5188-5194, Dec. 2008.
\bibitem{Li}
Y. Li, B. Bakkaloglu, and C. Chakrabarti, "A system level energy model and energy-quality evaluation for integrated transceiver front-ends," \emph{IEEE Trans. on VLSI Systems}, vol. 15, no. 1, pp. 90-103, Jan. 2007.
\bibitem{finite}
K. Tutuncuoglu and A. Yener, "Optimum transmission policies for battery limited energy harvesting nodes," \emph{IEEE Trans. Wireless Commun.}, vol. 11, no. 3, pp. 1180-1189, Mar. 2012.
\bibitem{elza2}
O. Orhan, D. Gunduz and E. Erkip, "Optimal packet scheduling for an energy harvesting transmitter with processing cost," \emph{Proc. IEEE Int'l Conf. on Communications (ICC)}, Budapest, Hungary, Jun. 2013.
\bibitem{Yang2012}
J. Yang and S. Ulukus, "Optimal packet scheduling in an energy harvesting communication system," \emph{IEEE Trans. on Commun.}, vol. 60, no. 1, pp. 220-230, Jan. 2012.
\bibitem{fade}
O. Ozel, K. Tutuncuoglu, J. Yang, S. Ulukus, and A. Yener, "Transmission with energy harvesting nodes in fading wireless channels: Optimal policies," \emph{IEEE JSAC,} vol. 29, no. 8, pp. 1732-1743, Sep. 2011.
\bibitem{broad}
O. Ozel, J. Yang, and S. Ulukus, "Optimal broadcast scheduling for an energy harvesting rechargeable transmitter with a finite capacity battery," \emph{IEEE Trans. Wireless Commun.}, vol. 11, no. 6, pp. 2193-2203, Jun. 2012.
\bibitem{broad2}
M. A. Antepli, E. Uysal-Biyikoglu, and H. Erkal, "Optimal packet scheduling on an energy harvesting broadcast link," \emph{IEEE JSAC,} vol. 29, no. 8, pp. 1721-1731, Sep. 2011.
\bibitem{yener2}
K. Tutuncuoglu and A. Yener, "Sum-Rate Optimal Power Policies for Energy Harvesting Transmitters in an Interference Channel," \emph{JCN Special Issue on Energy Harvesting in Wireless Networks,} vol. 14, no. 2, pp. 151-161, Apr. 2012.
\bibitem{deniz}
D. Gunduz and B. Devillers, "Two-hop communication with energy harvesting," \emph{Proc. CAMSAP,} San Juan, PR, Dec. 2011.
\bibitem{oner}
O. Orhan and E. Erkip, "Optimal transmission policies for energy harvesting two-hop networks," \emph{Proc. CISS,} Princeton, NJ, Mar. 2012.
\bibitem{deniz2}
B. Devillers and D. Gunduz, "A general framework for the optimization of energy harvesting communication systems with battery imperfections," \emph{Journal of Commun. and Netw., Spec. Issue on Energy Harvesting in Wireless Netw.,} vol. 14, no. 2, pp. 130-139, Apr. 2012.
\bibitem{kaya}
K. Tutuncuoglu and A. Yener, "Optimal power policy for energy harvesting transmitters with inefficient energy storage," \emph{Proc. CISS,} Princeton, NJ, Mar. 2012.
\bibitem{Jing}
J. Lei, R. Yates, and L. Greenstein, "A generic model for optimizing single-hop transmission policy of replenishable sensors," \emph{IEEE Trans. Wireless Commun.}, vol. 8, no. 2, pp. 547-551, Feb. 2009.
\bibitem{Kashef}
Z. Mao, C. E. Koksal, and N. B. Shroff, "Near optimal power and rate control of multi-hop sensor networks with energy replenishment: Basic limitations with finite energy and data storage," \emph{IEEE Trans. on Automatic Control,} vol. 57, no. 4, pp. 815-829, Apr. 2012.
\bibitem{Blasco-Gunduz-Dohler}
P. Blasco, D. Gunduz and M. Dohler, "A learning theoretic approach to energy harvesting communication system optimization," \emph{IEEE Trans. Wireless Commun.}, vol. 7, no. 31, pp. 1331-1341, Jul. 2013.
\bibitem{Gunduz-ComMag}
D. Gunduz, K. Stamatiou, N. Michelusi and M. Zorzi, "Designing intelligent energy harvesting communication systems," \emph{IEEE Communications Magazine}, vol. 52, no. 1, pp. 210-216, Jan. 2014.
\bibitem{process}
J. Xu and R. Zhang, "Throughput optimal policies for energy harvesting wireless transmitters with non-ideal circuit power," \emph{IEEE JSAC,} vol. PP, no. 99, pp. 1-11, May 2013.
\bibitem{nossek}
Q. Bai, J. Li, and J. A. Nossek, "Throughput maximizing transmission strategy of energy harvesting nodes," \emph{Proc. IWCLD,} Rennes, France, Nov. 2011.
\bibitem{Gregori}
M. Gregori and M. Payaro, "Throughput maximization for a wireless energy harvesting node considering the circuitry power consumption," \emph{Proc. IEEE Vehicular Technology Conference (VTC Fall)}, pp. 1-5, Sep. 2012.
\bibitem{elza}
O. Orhan, D. Gunduz, and E. Erkip, "Throughput maximization for an energy harvesting communication system with processing cost," \emph{Proc. IEEE Information Theory Workshop (ITW)}, Lausanne, Switzerland, Sep. 2012.
\bibitem{Boyd}
S. Boyd and L. Vandenberghe, \emph{Convex Optimization}, Cambridge University Press, 2004.
\bibitem{dynamic}
D. P. Bertsekas, \emph{Dynamic Programming and Optimal Control}, Athena Scientific, 2007.
\bibitem{infotheory}
T. M. Cover and J. Thomas, \emph{Elements of Information Theory}. New York: John Wiley and Sons Inc., 2006.
\end{thebibliography}
\end{document}